\documentclass[acmsmall,screen,nonacm]{acmart}
\usepackage{amsmath,amsfonts}
\usepackage{graphicx}
\usepackage{pifont}
\usepackage{pdfpages}
\usepackage{subcaption}
\usepackage{wrapfig}
\usepackage{pgfplots}
\usetikzlibrary{patterns}
\newtheorem{myDef}{Definition}
\newtheorem{myTheo}{Theorem}

\definecolor{codegreen}{rgb}{0,0.6,0}
\definecolor{codegray}{rgb}{0.5,0.5,0.5}
\definecolor{codepurple}{rgb}{0.58,0,0.82}
\definecolor{backcolour}{rgb}{0.95,0.95,0.92}
\colorlet{mygray}{black!30}
\colorlet{mygreen}{green!60!blue}
\colorlet{mymauve}{red!60!blue}
\usepackage{listings}
\usepackage[ruled,linesnumbered]{algorithm2e}

\AtBeginDocument{%
  }

\setcopyright{acmlicensed}
\copyrightyear{2024}
\acmYear{2024}
\acmDOI{XXXXXXX.XXXXXXX}

\acmConference[ISSTA]{ ACM SIGSOFT International Symposium
on Software Testing and Analysis}{June 03--05,
  2024}{Trondheim, Norway}
\acmISBN{978-1-4503-XXXX-X/18/06}

\begin{document}

\title{LoopSCC: Towards Summarizing Multi-branch Loops within Determinate Cycles}

\author{Kai Zhu}
\email{zhukai2022@iie.ac.cn}
\affiliation{%
  \institution{Institute of Information Engineering}
  \city{Beijing}
  \state{Beijing}
  \country{China}
}

\author{Chenkai Guo}
\affiliation{%
  \institution{Nankai University}
  \city{Tianjin}
  \state{Tianjin}
  \country{China}
}

\author{Kuihao Yan}
\affiliation{%
  \institution{Institute of Information Engineering}
  \city{Beijing}
  \state{Beijing}
  \country{China}
}

\author{Xiaoqi Jia}
\affiliation{%
  \institution{Institute of Information Engineering}
  \city{Beijing}
  \state{Beijing}
  \country{China}
}

\author{Haichao Du}
\affiliation{%
  \institution{Institute of Information Engineering}
  \city{Beijing}
  \state{Beijing}
  \country{China}
}

\author{Qingjia Huang}
\affiliation{%
  \institution{Institute of Information Engineering}
  \city{Beijing}
  \state{Beijing}
  \country{China}
}

\author{Yamin Xie}
\affiliation{%
  \institution{Institute of Information Engineering}
  \city{Beijing}
  \state{Beijing}
  \country{China}
}

\author{Jing Tang}
\affiliation{%
  \institution{Institute of Information Engineering}
  \city{Beijing}
  \state{Beijing}
  \country{China}
}


\begin{abstract}
Analyzing programs with loops is a challenging task,  suffering from potential issues such as indeterminate number of iterations and exponential growth of control flow complexity. 
Loop summarization, as a static analysis method for concrete semantic interpretation, receives increasing focuses in the field of loop program analysis. 
By analyzing and representing the regularity in loop control flow, it produces symbolic expressions semantically equivalent to the loop program, enhancing the accuracy and efficiency of loop analysis.
However, current loop summarization methods are only suitable for single-branch loops or multi-branch loops with simple cycles, without supporting complex loops with irregular branch-to-branch transitions.
In this paper, we proposed LoopSCC, a novel loop summarization technique, to achieve concrete semantic interpretation on complex loop control flow. 
LoopSCC first utilizes an inside-out transformation to convert the nested loop into non-nested one. 
Then, it analyzes the control flow at the granularity of single-loop-path and applies the strongly connected components (SCC for short) for contraction and simplification, resulting in the contracted single-loop-path graph (CSG for short). Based on the control flow information provided by the CSG, we can convert the loop summary into a combination of SCC summaries. 
When an SCC contains irregular branch-to-branch transitions, we propose to explore a convergent range to identify the determinate cycles of different execution paths, referred as \textit{oscillatory interval}. 
According to the analysis of oscillatory interval, the loop summarization composed of both iteration conditions and execution operations can eventually be derived recursively.
Extensive experiments compared to six state-of-the-art loop interpretation methods are conducted to evaluate the effectiveness of LoopSCC. From the results, LoopSCC outperforms comparative methods in both interpretation accuracy and application effectiveness.
Especially, LoopSCC achieves a 100\% interpretation accuracy on public common-used benchmark. In addition, a systematical study for loop properties on three large-scale programs illustrates that LoopSCC presents outstanding scalability for real-world loop programs. 
The LoopSCC and experimental data are available at \textit{https://anonymous.4open.science/r/LoopSCC-386F}.

\end{abstract}




\maketitle

\section{Introduction}
Dominant software engineering analysis techniques, such as symbolic execution \cite{baldoni2018survey} and model checking \cite{clarke1997model}, require simulating the execution of each reachable path within the target program. In this process, complex program structures, represented by \textit{loops}, raise significant challenges to execution-based analysis techniques.
For instance, when dealing with complex loop structures, both symbolic execution and model checking suffer from serious \textit{path explosion} where an infinite number of paths derived from the loop need to be executed, resulting in program crashes and unreasonable analysis. 

To address this challenge, a straightforward method is \textit{iteration limit}, which limits the number of loop iterations, simulating the loop as a finite path and executing it sequentially \cite{biere1999symbolic} \cite{healy1998bounding}. However, such method inevitably results in significant information loss, leading to serious biases in both execution and analysis. 
Building on this, Saxena et al. \cite{saxena2009loop} has proposed \textit{loop-dependent code summarization}, leveraging efficient static analysis techniques to obtain the \textit{runtime properties} (particularly the execution results) of loops, without actual loop execution.
Such static loop analysis methods can be classified into two main categories based on the semantic interpretation of the original code structure: \ding{182} \textit{Abstract interpretation} designs new program structures to approximate the target loop logic of original program, where the abstract semantics interpreted by the newly designed structures that is a superset of the original semantics, ensuring that all the reachable states of the original loops are covered by the abstract interpretation. 
\ding{183} \textit{Concrete interpretation} designs a computable mathematical model to interpret the program semantics of target loop logic in an accurate way, making the target program logic semantically equivalent to the designed model. In comparison, abstract interpretation is simpler to design and implement, offering a variety of variants on either structures or logic. Consequently, current loop analysis efforts \cite{ball2001automatic} \cite{kincaid2017compositional} tend to be built on the abstract interpretation. 
Nevertheless, \textit{abstract interpretation} fails to fully represent the semantics of the original loop structures, as it suffers from similar information loss or redundancy as the aforementioned \textit{iteration limit} method, leading to incomplete program analysis. 

On the contrary, \textit{concrete interpretation} achieves comprehensive summarization of loops. 
As stated by the \textit{Rice theorem} \cite{rice1953classes} and the \textit{halting problem} \cite{davis2013computability}, the computation of concrete semantics is proved as an \textit{undecidable problem}, meaning that only certain types of loop structures allow for concrete semantic interpretation.
For instance, representative \textit{concrete interpretations} \cite{saxena2009loop} \cite{godefroid2011automatic} tend to concentrate on single-branch loops since they are less affected by the \textit{undecidable computation}.
For multi-branch loops that involve irregular jumps between loop blocks, the concrete interpretation becomes exceptionally challenging as the undecidable program execution.
Aiming at this challenge, formal-method based efforts \cite{xie2016proteus}\cite{xie2017automatic}\cite{blicha2022summarization} provide valuable attempts at concrete semantic interpretation for the multi-branch loops, where specialized path structures, such as path dependency automaton (PDA), are proposed to capture the execution dependency between the paths and transform irregular loop paths into parameterized periodic iterations. By analyzing the parameter expressions (such as iteration counters), they are able to interpret the periodic execution of the proposed path structures and further produce the semantic summarization of the complex multi-branch loops. 
However, such parameterized periodic summarization often fails to satisfy the requirements for \textit{parameter inductiveness}, known as the \textit{inductiveness trap}, which leads to significant uncertainty in the summarized iteration cycles, blocking the interpretive computation of loop execution.

In this work, instead of the traditional \textit{parameterized periodic summarization}, we propose a novel loop summarization with \textit{determinate cycles} to explicitly interpret the logical semantics within multi-branch loops, and further build a practical analysis framework \textit{LoopSCC} for precise, efficient and generalized analysis of program semantic.
First, to facilitate the summarization, the LoopSCC converts the target loop into a \textit{canonical form} with single input and output, using an existing \textit{Gaussian Elimination} based method \cite{Ashcroft_Manna_1979} \cite{Ammarguellat_1992}.
Then, based on the transitions among the blocks within loops, we construct a \textit{SPath graph} to represent the fine-grained control flow of the loops. By simplifying the nodes in the SPath  graph at the granularity of  \textit{strong connected components} (\texttt{SCC} for short), we can further obtain a directed acyclic graph focused on the SCC, referred to as contracted SPath graph (\texttt{CSG} for short).
For a target program with a complex loop structure, the execution will iterate repeatedly inside SCC and possibly exhibit a certain periodicity when iterating sufficiently. 
To extract such periodicity, we proposed \textit{oscillatory interval} to represent the iterations of SCC into a piecewise function calculation within a limited value scope.
To determine the \textit{oscillatory interval} within the target loop execution, we have proposed an iterative \textit{search algorithms}.
After that, the LoopSCC utilizes function operations such as addition and subtraction to extract the periodicity in the oscillatory interval. 
In particular, LoopSCC uses the pigeonhole principle \cite{trybulec1990pigeon} to derive the periodicity of discrete values directly. 
Finally, the target loop can be summarized by computing the result of periodic function extracted from the oscillatory interval.

We evaluate the effectiveness of LoopSCC from different perspectives through extensive experiments.
Firstly, we evaluate the summarization precision of LoopSCC compared with state-of-the-art baselines on public datasets, where LoopSCC achieves a 100\% summarization accuracy, outperforming all the baselines.
Secondly, we performed program verification using the benchmark \textit{SV-COMP 2024}.
The results indicate that LoopSCC correctly verifies 86\% of the test cases, outperforming the best competing tool VeriAbsL \cite{darke2023veriabsl} by 10.3\%.
Thirdly, LoopSCC is integrated into typical program analysis tools
to test the support of code analysis. The results demonstrate that LoopSCC significantly improves the analysis efficiency and coverage.
Finally, we systematically investigated the feasibility of using the LoopSCC to summarize loops with non-memory-related operations in three large open-source programs: \texttt{Bitcoin}, \texttt{musl}, and \texttt{Z3}. 
The results indicate that 81.5\% of the loops can be summarized using the LoopSCC, highlighting its outstanding scalability.

In summary, this work makes the following contributions:
\begin{itemize}
\item We proposed LoopSCC, a novel loop summarization framework based on strongly connected components, along with a dynamic programming-based interpretation algorithm to handle the implicit relationships within SPath conditions.
\item We proposed the \textit{finite oscillatory interval} for the code execution within loop structures, and conducted an in-depth analysis of the periodic variation patterns of the oscillatory interval, along with a concrete interpretation and computation scheme for the execution outcomes.
\item We conducted extensive comparative experiments on public datasets against state-of-the-art loop summarization methods, and the results demonstrate that the proposed LoopSCC is not only theoretically sound and effective but also significantly enhances practical code analysis techniques.
\end{itemize}

\section{Motivation}\label{sec:motivation}

\begin{figure}[!htbp]
\begin{minipage}[t]{.2\textwidth}
\begin{lstlisting}
while i < 100:
    # Branch A
    if x > 1:
        x += 1
        i += 3
    # Branch C
    elif x < -1:
        x += 1
        i += 5
    # Branch B
    else:
        x += 1
        i += 7
\end{lstlisting}
\subcaption{Acyclic Multi-branch Loop}\label{a}
\end{minipage}
\hfill
\begin{minipage}[t]{.2\textwidth}
\begin{lstlisting}
while i < 100:
    # Branch A
    if x > 1:
        x -= 1
        i += 3
    # Branch C
    elif x < -1:
        x += 1
        i += 5
    # Branch B
    else:
        x += 1
        i += 7
\end{lstlisting}
\subcaption{Cyclic Multi-branch Loop}\label{b}
\end{minipage}
\hfill
\begin{minipage}[t]{.2\textwidth}
\begin{lstlisting}
while i < 100:
    # Branch A
    if x >= 50:
        x -= 2
        i += 3
    # Branch C
    elif x < 0:
        x += 1
        i += 5
    # Branch B
    else:
        x += 11
        i += 7
\end{lstlisting}
\subcaption{Inductiveness Trap}\label{c}
\end{minipage}
\hfill
\begin{minipage}[t]{.2\textwidth}
\begin{lstlisting}
while i < 100:
    # Branch A
    if x > 1:
        x -= 5
        i += 3
    # Branch C
    elif x < -1:
        x += 1
        i += 5
    # Branch B
    else:
        x += 9
        i += 7
\end{lstlisting}
\subcaption{Connected Cycles}\label{d}
\end{minipage}
\caption{Motivating Examples.}
\label{fig:samples}
\end{figure}

\begin{figure}[htbp]
\begin{minipage}[b]{.2\textwidth}
\includegraphics[width=\textwidth]{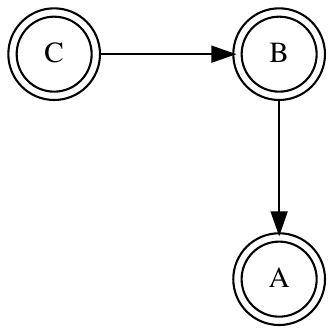}\\
\subcaption{}
\end{minipage}
\hfill
\begin{minipage}[b]{.2\textwidth}
\includegraphics[width=\textwidth]{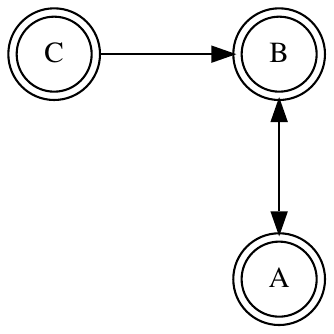}\\
\subcaption{}
\end{minipage}
\hfill
\begin{minipage}[b]{.2\textwidth}
\includegraphics[width=\textwidth]{Img/sample2}\\
\subcaption{}
\end{minipage}
\hfill
\begin{minipage}[b]{.2\textwidth}
\includegraphics[width=\textwidth]{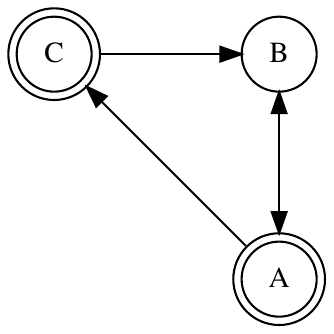}\\
\subcaption{}
\end{minipage}
\caption{CFGs Corresponding to Motivating Examples.}
\label{fig:jump}
\end{figure}

Fig. \ref{fig:samples} lists five typical loop slices that motivate our work. 
Fig. \ref{fig:samples}(a) presents a simple and generic three-branch ($A, B$, and $C$) loop slice. Since there are no \textit{cyclic jump} among the three branches (as shown in Fig. \ref{fig:jump}(a)), they essentially correspond to single-branch loops and can be correctly summarized by traditional methods such as \textit{LESE} \cite{saxena2009loop}, \textit{APLS} \cite{godefroid2011automatic} , \textit{APC} \cite{strejvcek2012abstracting} , \textit{Proteus}\cite{xie2016proteus} and \textit{Wsummarizer}\cite{blicha2022summarization}. For branch $A$, assuming the initial states of variables $x$ and $i$ are $x_0$ and $i_0$, the summary for $A$ is $x = x_0 + N \cap i = i_0 + 3*N$, where $N$ is the number of iterations of $A$.

However, in Fig. \ref{fig:samples}(b), branch $A$ and $B$ form a simple \textit{cyclic jump}, which cannot be handled by traditional single-branch methods \cite{saxena2009loop} \cite{godefroid2011automatic}. 
Unlike Fig. \ref{fig:samples}(a), the number of iterations for branch $A$ is $x_0 - 1$, and the states of $x$ and $i$ become $1$ and $i_0 + 3x_0 - 3$, respectively after iterations. Subsequently, branch $B$ is executed \textit{once}, after which the states of $x$ and $i$ change to $2$ and $i_0 + 3x_0 + 2$, respectively. 
In the example, the number of iterations for $A \leftrightarrow B$ cycle can be represented as a parameterized expression $x_0$ along with the initial variables.
Existing efforts (\textit{Proteus}\cite{xie2016proteus} and \textit{Wsummarizer}\cite{blicha2022summarization}) transform the parameterized expression into a single-branch loop for summarization.

However, such parameterized expression transformation does not always work in the loop summarization. For instance, in the loop operation of Fig. \ref{fig:samples}(c), a slight change in the operation on $x$ requires the iteration count $N$ of branch $B$ to satisfy the conditions $(11*(N-1) + x < 50) \cap (11*N + x \geq 50)$, such that the value of $N$ can be derived as $\lfloor\frac{(60-x)}{11}\rfloor$. 
The derived value is a \textit{non-inductive variable} \cite{xie2016proteus} which cannot be directly analyzed and computed, and thus are hard to be summarized by traditional methods. Such problem is called \textit{inductiveness trap} in our work, meaning that simple operations like addition, subtraction, multiplication, and division can generate complex and non-inductive mathematical functions, such as $floor()$, exponentiation, logarithm, and their combinations, when summarizing the number of loop iterations.

Additionally, multiple branches within the loop of Fig. \ref{fig:samples}(d) may evolve in \textit{connected-cycle jumps}, as shown in Fig. \ref{fig:jump}(d). When a branch ($A$ or $C$) completes its iterations, it is hard to determine the destination branch for jumping, which makes the current methods ineffective.

The motivating examples illustrate that the current loop summarization efforts only cover single-branch loops or simple multi-branch loops, which is hard to 
address the commonly occurring issues like \textit{inductiveness trap} and \textit{connected cycle}. That inspires us to explore a new avenue to 
address these issues fundamentally by extracting specific \textit{determinate cycles} within complex multi-branch loops.

\section{Preliminaries}\label{sec:prelimi}
To facilitate the description of our approach, we will first clarify some existing related concepts on the \textit{graph theory} \cite{bondy2008graph}.

\begin{myDef}
	\label{def1}
   \textbf{Cycle.} Given a directed graph $G = (N, E)$, a cycle $C=\{c_1,c_2,\ldots,c_k\}\subseteq N$ is a set of nodes satisfying that $\forall i \in [1,k], \forall j \in [1,k]: c_i \to c_j$, where $\to$ denotes the reachable relationship; $N$ and $E$ are the set of nodes and edges, respectively. 
\end{myDef}
Any two nodes in a cycle are reachable to each other.

\begin{myDef}
	\label{def2}
   \textbf{Strongly Connected Component (SCC).}  SCC $scc={n^{scc},e^{scc}} \subseteq G$ is a cycle in a directed graph $G$., and no additional node can be added to maintain the cycle, where $n^{scc}$ and $e^{scc}$ are the set of nodes and edges of $scc$, respectively. 
\end{myDef}
The number of nodes $|n^{scc}|$ in an SCC is referred to \textit{the order of SCC}. A 1-order SCC is called a \textit{self-loop}, while a 0-order SCC contains a single node without any edge. SCC with an order greater than 1 is called \textit{high-order SCC}; otherwise, it is called \textit{low-order SCC}.

\begin{myDef}
	\label{def3}
   \textbf{Contraction.} A directed graph $G$ can be partitioned  into multiple classes by the SCCs , i.e., $G = {N,E}=\{scc_1, scc_2, ..., scc_k\}$, where $N= \bigcap_{i=1}^k n^{scc}_i$ and $E= E_{inner} \cup E_{inter}=  (\bigcup_{j=1}^k e^{scc}_j) \cup E_{inter}$.  The \textit{contraction} is to construct a graph $G'=\{N', E'\}$ that is semantically equivalent to $G$, so that $N'=\{n_1, n_2, ...,n_k\}$ and $E'= E_{inter}$.
\end{myDef}
$E_{inner}$ refers to the inner edges of SCCs and $E_{inter}$ refers to the set of edges between SCCs. Essentially, the contraction is the process of simplifying each strongly connected class into a single abstract node,
serving as the fundamental operation in further cycle summarization. It is obvious that the new directed graph $G'$ after contraction is acyclic.
Then, we clarify the concepts related to a single loop structure. Due to the complexity of loop structures, it is difficult to handle them uniformly. We first define an easily processed standard loop structure called the \textit{canonical form}. 
\begin{myDef}
	\label{def4}
   \textbf{Canonical Form}. The canonical form $CF$ of a single loop iteration is a \textit{directed acyclic graph}, which has a unique extry and exit, i.e.,  $CF(N, E, x) \Rightarrow y$, where $N$ and $E$ are the set of nodes and edges contained by the graph; $x$ and $y$ are the unique extry and exit, respectively.
\end{myDef}

It can be inferred from existing structured programming theorems (\textit{Böhm-Jacopini theorems}) \cite{Böhm_Jacopini_1966}\cite{Kozen_Tseng_2008}, complex single loop structure that contains multiple entries or exits (e.g., impacted by \texttt{break} or \texttt{goto} sentence) can be transformed to the \textit{canonical form}. In practice, we have implemented a generalized transformation module employing traditional program normalization algorithms 
\cite{Ashcroft_Manna_1979} \cite{Ammarguellat_1992}, and integrated it into the LoopSCC.
Afterwards, we explore the paths within the canonical form, denoted as \textit{single-loop pat}h.

\begin{myDef}
	\label{def5}
   \textbf{Single-loop Path (SPath)}. SPath $sp=\{n_1\to n_2\to \ldots \to n_k\}$ is a node sequence in the canonical form $CF(N, E, x)$, where $\{n_1,n_2,\ldots,n_k\} \subseteq N $.
\end{myDef}

\begin{myDef}
	\label{def6}
   \textbf{SPath Operation}. Given a SPath $sp$, the operation set $sp.OP$ on the $sp$ is a map from the values of prefix variables $sp.Pre$ to the values of post variables $sp.Post$ of $sp$, denoted as $sp.Op: sp.Pre \to sp.Post$.
\end{myDef}

Note that the memory-oriented operations, e.g., array indexed by a variable or pointer assignment, are out of our scope, since the objects of such operations are uncertain, which randomizes the direction of loop iteration.

\begin{algorithm}[ht]
\SetKwInOut{Input}{Input}\SetKwInOut{Output}{Output}
\caption{Computation of SPath Condition and Operation.}\label{alg:SPath-Condition}
\Input{$sp$}
\Output{$sp.Cond$, $sp.Op$}

$sp.Op \gets sp.Pre$\;

\ForEach{$node \in sp.Nodes$}{
    \uIf{$node$ is Conditional Node}{
        $sp.Cond.insert(Substitude(node.Cond,sp.Op))$
    }
    
    \Else{
        $new\_value \gets Substitude(node.Op.rvalue,sp.Op)$\;
        $sp.Op[node.Op.lvalue] \gets new\_value$\;
    }
}
\end{algorithm}

\begin{myDef}
	\label{def7}
   \textbf{SPath Jump}. Given two SPaths $sp_1$ and $sp_2$, if the post variable values of $sp_1$, i.e., $sp_1.Post$ make the condition set of $sp_2$, i.e.,  $sp_2.Cond$ hold, there is a \textit{jump} between $sp_1$ and $sp_2$, denoted as $\mathcal{J}(sp_1 ,sp_2)$.
\end{myDef}

To determine the jumps of SPath in the loop iteration, it requires computing the SPath condition and operation first. 
Algorithm \ref{alg:SPath-Condition} provides a method for calculating the loop conditions and operations of the SPath through a forward traversal.
Afterwards, we define \textit{SPath graph} to build the abstract structure for an entire loop.

\begin{myDef}
	\label{def8}
   \textbf{SPath Graph}. An SPath graph $SG_l$ for a loop $l$ is a quadruple, i.e., $\{sp_s, sp_e, SP, \mathcal{J}\}$, where $sp_s$ and $sp_e$ are two empty SPaths representing the starting node and end node; $SP$ is the set of SPaths in the $l$; $\mathcal{J}$ denotes the set of SPath jumps in the $l$.
\end{myDef}

In the SPath graph, the SPaths and jumps in the loop are abstracted as nodes and edges, respectively. Then, the execution process of the loop is the sequence of SPaths from the $sp_s$ to $sp_e$, i.e., ${sp_s, sp_1, sp_2, ...,  , sp_e}$.
The SPath graph may contain cycles which complicates the further summarization. To alleviate it, LoopSCC transforms the original SPath graph into a directed acyclic graph based on the SCC-based contraction as Def.\ref{def3}, called \textit{CSG} (Contracted SPath graph). Furthermore, redundant edges that are not part of the path from the $sp_s$ to $sp_e$ are removed.

\begin{figure}[ht]
\begin{minipage}[b]{.4\textwidth}
\begin{lstlisting}
while i < 100:
    # Branch A
    if x < 0:
        x += 2
        i += 3
    else:
        x += 7
        # Branch B
        if x < 5:
            x += 3
            i += 1
        # Branch C
        else:
            x -= 12
            i += 2
\end{lstlisting}
\subcaption{Source Code}
\label{lst:sample1}
\end{minipage}
\hfill
\begin{minipage}[b]{.25\textwidth}
\includegraphics[width=\textwidth]{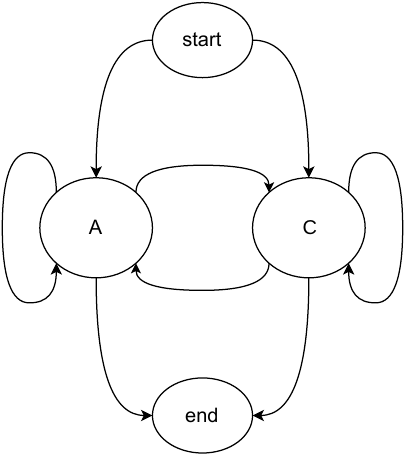}\\
\subcaption{SPath Graph}
\label{fig:pfg}
\end{minipage}
\hfill
\begin{minipage}[b]{.1\textwidth}
\includegraphics[width=\textwidth]{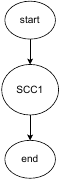}\\
\subcaption{CSG}
\label{fig:route}
\end{minipage}
\caption{Example of SPath Graph and CSG.}
\label{fig:SPathGraph-CSG}
\end{figure}

\textbf{Example.} We use the analysis of the loop program in Fig. \ref{fig:SPathGraph-CSG}(a) to illustrate the process of constructing the SPath graph and the CSG. 
We first compute the $sp.Op$ for \texttt{Branch B}. For initial variables $x_0, i_0$, after being fed into the statements related to \texttt{Branch B}, i.e., $x+=7; x+=3;$ and $i+=1;$, the $sp.Op$  can be computed as $x=x_0+10$ and $i=i_0+1$.
Then we compute the $sp.Cond$ for \texttt{Branch B}. 
There are two conditional statements related to \texttt{Branch B}, i.e., $x >= 0$ and $x<5$. For $x >= 0$, since the variable value is the same as the pre-variable, we directly obtain the conditional expression $x_0 >= 0$. For $x<5$, there exists a mapping relationship $x=x_0+7$ between the variable and the pre-variable, and thus the conditional expression is computed as $x_0+7<5$. Afterwards, $sp.Cond$ for \texttt{Branch B} is an unsatisfiable condition $x_0 >= 0 \cap x_0<-2$, which demonstrates the \texttt{Branch B} is an invalid SPath.

Similarly, we continue to compute the $sp.Op$ for \texttt{Branch A} is $x=x_0+2$ and $i=i_0+3$, the $sp.Cond$ is $x_0<0$; the $sp.Op$ for \texttt{Branch C} is $x=x_0-5$ and $i=i_0+2$, the $sp.Cond$ is $x_0>=0  \cap x_0>=-2$, i.e., $x_0>=0$. Since the $sp.Cond$ of the two branches do not conflict with the loop condition, both of them are valid SPaths, denoted as $sp_A$ and $sp_C$.

Subsequently, we compute the SPath jumps $\mathcal{J}$ between valid SPaths. Assuming that the $sp_A.Pre$ is $x_0$ and $i_0$; $sp_A.Post$ is $x_1$ and $i_1$, then we have $x_1=x_0+2$. Thus, $sp_A.Cond$ and $sp_C.Cond$ can be satisfied simultaneously, making $\mathcal{J}(sp_A,sp_C)$ exists. Similarly, $sp_A.Cond$ and the condition of end SPath $sp_e.Cond$ can be satisfied simultaneously, making $\mathcal{J}(sp_A,sp_e)$ exists.
After computing all the SPath jumps, the SPath graph can be built as Fig. \ref{fig:SPathGraph-CSG}(b). Finally, we can build the CSG by \textit{contraction} for SCCs, as shown in Fig. \ref{fig:SPathGraph-CSG}(c).

\section{SCC-based Loop Summarization} \label{sec:approach}

\subsection{Overview}

\begin{figure}[htbp]
    \centering
    \includegraphics[width=.85\textwidth]{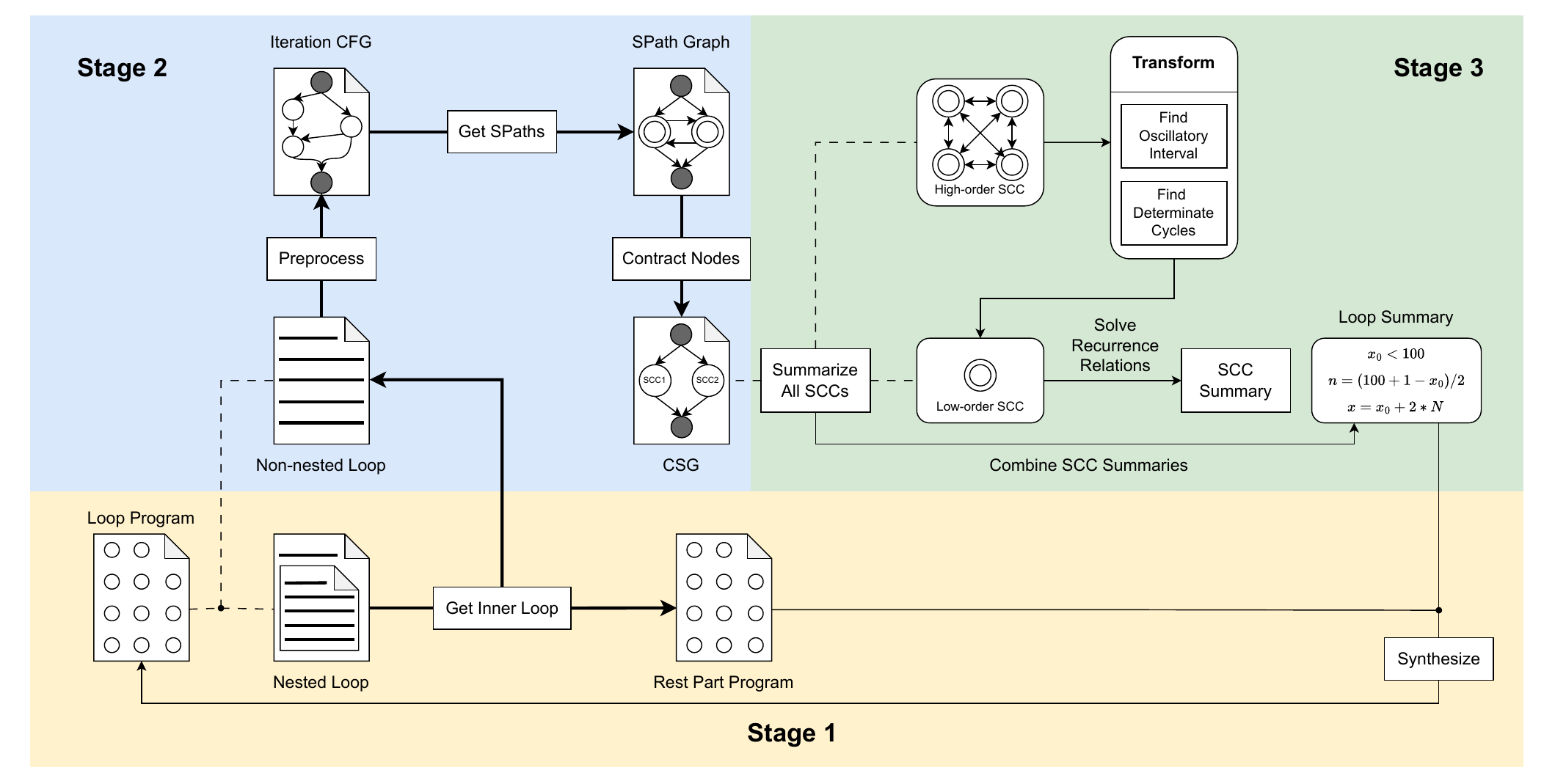}
    \caption{The Workflow of LoopSCC}
    \label{fig:overview}
\end{figure}
Fig. \ref{fig:overview} presents the workflow of LoopSCC. 
First, since the nested loops are prone to complicating the construction of SPath graph, LoopSCC adopts the inside-out conversion algorithm 
to transform the nested loops into non-nested ones (stage \ding{182}). 
Then LoopSCC builds the SPath graph and CSG for the target loop by SPath extraction, jump condition computing, and SCC-based contraction (stage \ding{183}).
Afterwards, the SCCs are comprehensively computed to generate the final loop summarization (stage \ding{184}), which can be divided into three parts according to the order of SCCs:
1) For 0-order SCCs, the loop summarization can be directly determined and obtained by corresponding SPath operations. 
2) For 1-order SCCs, the loop summarization requires solving the closed-form expression of recursive SPath operations. 
3) For high-order SCCs, we introduce \textit{oscillatory interval} and \textit{determinate cycles} to transform the high-order SCCs into 1-order ones, which are then subjected to a unified loop summarization.

\subsection{SCC-based Flow Analysis}
\subsubsection {Transformation of Nested Loop}
Nested loops can introduce complex data flow relationships, making direct summarization challenging. Therefore, for nested loops, we need to convert them into non-nested loops before summarization. Our approach is to first perform a loop summarization on the innermost loop, then transform it into a linear program and reintegrate it into the original nested loop program, thereby eliminating the nesting.
A typical loop summarization (for the summarization details, see the following \S\ref{subsec:loop-summ}) contains two parts: 1) Value range of involved variables; 2) Mappings between pre-variables and post-variables, which are transformed by LoopSCC separately. For the variable values in the summarization, LoopSCC transforms them into conditions for different branches based on the condition states. For variable mapping relationships, LoopSCC converts them into assignment statements within corresponding branches according to the operation types. 

For instance, if the innermost loop of a target nested loop is summarized as   $(x_0>=100 \cap x=x_0) \cup (x_0<100 \cap N=(100+1-x_0)/2 \cap x=x_0+2*N)$, LoopSCC can transform it into an equivalent non-nested linear program in two steps.
\ding{182} For $(x_0>=100 \cap x=x_0)$, the condition $(x_0>=100)$ is transformed to condition sentence  $if(x>=100)$, and the assignment $x=x_0$ is transformed to operation $x=x$ for the first branch. \ding{183} For $(x_0<100 \cap N=(100+1-x_0)/2 \cap x=x_0+2*N)$, it can be transformed into operations $N = (100 + 1 - x); x = x + 2 * N;$ with condition $if(x<100)$ for another branch.

\subsubsection{Construction of CSG}
For non-nested loops, we analyze the SCC related structures to facilitate further summarization. 
First, we preprocess the loop program and obtain the \textit{continuation equations} \cite{Ammarguellat_1992} using a \textit{depth-first search algorithm}. 
Subsequently, we solve the system of \textit{continuation equations} using a \textit{Gaussian elimination-like resolution} method. 
After transforming the results, we can obtain the canonical form of the loop.

To construct the SPath graph and CSG defined in \S\ref{sec:prelimi}, we then generate the entire control flow graph (CFG) for the target program by existing construction algorithm \cite{allen1970control}. 
Using the \textit{reverse post-order depth-first search algorithm} \cite{tarjan1974finding}, we identified all \textit{dominator nodes} in the CFG, which allowed us to generate all SPaths of the CFG through dominator connections. 
Meanwhile, the jump relations between SPaths can be calculated by the \textit{Z3 solver}. Subsequently, the SPath graph can be constructed by  the generated SPaths and jump relations. 
After that, all the SCCs within the SPath graph are explored with \textit{Tarjan}'s algorithm \cite{tarjan1972depth}.
Based on the jump relations of the SPaths, we can derive the jump relations of the SCCs and subsequently construct the CSG.

\subsection{Loop Summarization} \label{subsec:loop-summ}
\subsubsection{Summarization for 1-order SCC}
The summarization of the 1-order SCC can be treated as $n$ iterations of a SPath $sp$, that is, given the $sp_1.Pre$, the summarization is to compute the $sp_n.Post$, where $sp_i$ denotes the SPath of $i$-th iteration.
The essence is to get the closed-form expression of the recurrence relation for the iteration and compute the value of the $n$-th term of the expression. 
It can be observed that in $i$-th iteration ($i<n$), the $sp_i.Pre$ satisfies the $sp_i.Cond$ but $sp_{n-1}.Post$ does not satisfies the $sp_n.Cond$. 
LoopSCC translates such {existential quantification} statements into constraints and feeds them into an SAT solver to compute the satisfiable values of $n$.
Note that among the computed values, only the smallest one is the final $n$ we seek, which is hard to be explored directly. Therefore, we design an improved dynamic programming algorithm to find this value as shown in the Algorithm \ref{alg:dynamic-programming}.

\begin{algorithm}[ht]
\SetKwInOut{Input}{Input}\SetKwInOut{Output}{Output}
\caption{Determination of the Number of Iterations}\label{alg:dynamic-programming}
\Input{$sp.Cond$, solver}
\Output{the number of iterations in 1-order SCC}

$solver \gets \{n>0, sp_{n-1}.Cond, \neg sp_{n}.Cod) \}$\;
$n_{val} \gets solver.solve()$\;

\While{True}{
    $solver \gets n < n_{val}$\;
    
    \uIf{$solver.solve()$ is SAT}{
        break \;
    }
    
    \Else{
        $n_{val} \gets solver.solve()$\;
    }
    
}

\end{algorithm}

From Algorithm \ref{alg:dynamic-programming}, once a satisfiable value $n_{val}$ is found, an extra constraint $n<n_{val}$ is added to explore a smaller $n_{val}$. The process continues until no smaller $n_{val}$ can be found, where the actual number of iterations in 1-order SCC is determined.

\textbf{Example.} Review back to the SPath \texttt{A} in the example of Fig. \ref{fig:SPathGraph-CSG}, where the $sp.Cond$ and $sp.Op$ are $x<0$ and $x=x_0+2$, respectively. LoopSCC generates the closed-form expression $x=x_0+2*N$ and feeds the condition $x_0 + 2*N>=0 \cap x_0+2*(N-1)<0$ into the SAT solver. Subsequently, we obtain a satisfiable value of $n$ as $(1-x_0)/2$. Since there are no smaller satisfiable values, this value is the final number of iterations.

\textbf{Existential Quantification vs. Universal Quantification.}
Advanced loop summarization efforts, (e.g., \textit{Proteus} \cite{xie2016proteus} and \textit{Wsummarizer}\cite{blicha2022summarization}) proposed to extract the explicit symbolic representation for the number of iterations $n$, i.e., compute the $n$ from $sp_n.Cond$, which is a  \textit{universal quantification} method. 
However, practical conditions often contain \textit{implicit expression} that can not be directly computed, making existing efforts invalid. For instance, for a simple loop \texttt{while $x^7 < x^3 + 2$: $x = x + 2$}, whose $sp.Op$ and $sp.Cond$ are $x=x_0+2*n$ and $x^7<x^3+2$, respectively, the final condition can be referred as $(x_0+2*n)^7 < (x_0+2*n)^3 + 2$. This is a typical \textit{implicit expression}, from which an explicit expression regarding $n$ cannot be extracted. 
On the contrary, LoopSCC employs an existential quantification strategy, that is to identify an appropriate $n$ that satisfies the $sp_n.Cond$, effectively addressing the issues raised by the \textit{implicit expressions}. In fact, our summarization of 1-order SCC exhibits strong universality: as long as each iterative $sp_i.Cond$ based on the $sp_i.Op$ (the mapping from $sp_i.Pre$ to $sp_i.Post$) can be expressed as closed-form expressions, and the $sp_i.Cond$ can be solved by an SAT solver, LoopSCC can successfully summarize them.

\subsubsection{Summarization for High-order SCC}\label{subsec:high-order}

We first conduct a loop summarization of real addition and subtraction for 2-order SCCs, and then expand it into a general scenario.
Reviewing Fig. \ref{fig:samples}(c) in the motivation section, where branches $A$ and $B$ form a 2-order SCC. 
In the loop, when the value of $x$ is less than 50, $x$ continues to increase; otherwise, it decreases. Furthermore, the value of $x$ gradually converges and stabilizes within the interval $[48,61)$. When $x$ falls within this interval, the program repeatedly executes the \textit{if and else} branches, so we refer to this convergence interval as the \textit{oscillatory interval}. The operations within the oscillatory interval are piecewise functions, but we can represent them directly using a non-piecewise modular addition function, i.e., $x_n=(x_{n-1}-48+11)\bmod(61-48)+48$, which can be simplified to $x_n=(x_{n-1}+2)\bmod 13 +48$ whose closed-form expression is $x_n=(x_0+2*n)\bmod 13 +48$.
Therefore, in this case, the summarization of high-order SCCs within the oscillatory interval can be viewed as a summarization of 1-order SCCs. When the SCCs are outside the oscillatory interval, we can calculate how many iterations it takes to transition into the oscillatory interval, which essentially is also a summarization problem of 1-order SCCs.
This case inspires us a feasible avenue for the high-order summarization is to seek the \textit{oscillatory interval} and analyze its \textit{determinate periodicity} within the SCCs. First, we clarify some conceptions by the following definitions. 

\begin{myDef}
	\label{def9}
   \textbf{Oscillatory Interval.} Oscillatory interval $[a,b]$ is an enclosed interval under operations $sp.Op$. If a pre-variable of $sp_i$ satisfies $sp_i.pre \in [a,b]$, after $sp_i.Op$, the corresponding post-variable still satisfies $sp_i.Post \in [a,b]$.
\end{myDef}

\begin{myDef}
	\label{def10}
   \textbf{Execution Periodicity.} In a high-order SCC, if there exists a period length $T$ such that for any $i\in [m,n]$, making $sp_i=sp_{i+T}$ hold, we call the SCC in $[m,n]$ has execution periodicity, denoted as $\widetilde{scc}_{[m,n]}^T$, and the $[m,n]$ is called node interval. 
\end{myDef}

Any node interval $[m,n]$ corresponds to a value interval $[a, b]$ of involved variables in the $[m,n]$, i.e.,  $\forall sp \in \widetilde{scc}_{[m,n]}^T, sp.Post \in [a, b]$. Thus, we can seek the \textit{oscillatory interval} by merging the value intervals of node intervals. To facilitate the implementation, we categorize the value intervals of SPaths during the SCC execution into three types according to the triggered destination.

\begin{itemize}
\item \textit{Jumping Interval (J-Interval)}: The value interval that triggers the SPath node to jump to another node in the current SCC. 
\item \textit{Iteration Interval (I-Interval}): The value interval that triggers the SPath node to jump to itself in the SCC. 
\item \textit{Terminal Interval (T-Interval)}: The value interval that triggers the SPath node to jump to the node in another SCC, which indicates the execution terminal of the current SCC.
\end{itemize}

\begin{myTheo}\label{theo1}
If the oscillatory interval $O$ of a high-order SCC, $O$ contains all the J-Intervals and can be divided into a finite number of periodic subintervals,, then the summarization of $scc$ can be converted into a summarization of low-order SCC.
\end{myTheo}

\begin{proof}\label{proof1}
The entire execution of $scc$ contains two typical parts according to the value intervals of involved variables, i.e., the inner of $O$ and the outer of $O$.
\ding{182} For the inner one, the SPath sequence triggered by by the periodic subintervals of $O$ is also periodic and can be summarized according to the 1-order SCC summarization.
Therefore, summarizing the inner of $O$ results in the combined summary of all periodic subintervals.
\ding{183} For the outer one, there are only I-Intervals and T-Intervals for the execution since all the J-Intervals are contained by $O$. The SCC summarization triggered by I-Intervals and T-Intervals can refer to 1-order and 0-order SCCs, respectively.
\end{proof}

Therefore, the key to summarizing $scc$ is to \textit{identify $O$} that covers all J-Intervals of $scc$, as well as to \textit{determine the periodicity} of $O$: 
\ding{182} To identify $O$, we propose an \textit{iterative convergence algorithm} illustrated in the Algorithm \ref{alg:oscillatory}, which starts from all nodes triggered by J-Intervals, and continuously explores the remaining nodes until the values of all explored nodes converge to a certain interval. 
\ding{183} The periodicity of $O$ is determined by Theorem \ref{theo2}.

\begin{algorithm}[ht]
\SetKwInOut{Input}{Input}\SetKwInOut{Output}{Output}
\caption{Identification of Oscillatory Interval}\label{alg:oscillatory}
\Input{$S$: all nodes of target SCC; $J$: all J-Nodes in $S$; $OP()$: operation set of $S$}
\Output{$O$: oscillatory interval}

$A \gets J$\;
$B \gets None$\;

\While{True}{
    $B \gets OP(A) \cap S$\;
    
    \uIf{$(B\nsubseteq A)$}{
        $A \gets A \cup B$ \;
    }
    
    \Else{
        $O \gets A$\;
        break\;
    }
}
\end{algorithm}

\begin{myTheo}\label{theo2}
If an oscillatory interval $O$ of a high-order SCC, $O$ contains finite $N$ values, it has at most $N$ periodic subintervals for the execution of $scc$.
\end{myTheo}
\begin{proof}\label{proof2}
According to the enclosed operations of $O$ and the \textit{pigeonhole principle} \cite{trybulec1990pigeon}, after arbitrary $N$ times of executions from the nodes indicated by the $N$ values within $O$, two identical nodes will be visited, resulting in the same execution path sequence with a \textit{period interval}. Consequently, there exist at most $N$ \textit{period intervals} in the $O$.
\end{proof}
Note that from Algorithm \ref{alg:oscillatory}, the seeking of the oscillatory interval is a general process oriented toward SCCs and is unrelated to whether the SCC contains \textit{connected cycles} referred in the \S\ref{sec:motivation}, allowing LoopSCC's loop summarization to cover connected cycles.

\textbf{Example.} We use a typical 2-order SCC loop in Fig. \ref{2orderscc}(a) to describe the summarization based on the oscillatory interval $O$. In the loop, \texttt{SPath A} and \texttt{SPath B} have J-Interval of $[3,5)$ and $[5,10)$, respectively, and all the rest value intervals belong to I-Intervals. To achieve $O$, the initial value interval $A$ is set as all the J-Intervals, i.e., $A= [3,5) \cup [5,10)= [3,10)$. After SPath operation, the value interval $B$ is $[0,7)$. At this time, $B \nsubseteq A$ holds, so $A$ is extended to $A \cup B = [0,10)$ and continues to be operated. After the second operation, the value interval $B$ is still $[0,7)$ and does not satisfy $B \nsubseteq A$. Then we get $O= A= [0,10)$. 

\begin{figure}[ht]
\begin{minipage}[b]{.3\textwidth}
\begin{lstlisting}
while i < 100:
    if x < 5:
        x = x + 2
        i = i + 3
    else:
        x = x - 5
        i = i + 4
\end{lstlisting}
\subcaption{}
\end{minipage}
\hspace{5em}
\begin{minipage}[b]{.3\textwidth}
\begin{lstlisting}
while i < 100:
    if x < 4:
        x = x + 2
        i = i + 3
    else:
        x = 2 * x - 8
        i = i + 4
\end{lstlisting}
\subcaption{}
\end{minipage}
\caption{Examples of the loop contains 2-order SCC}
\label{2orderscc}
\end{figure}

After that, LoopSCC summarizes the loop upon the computed $O$. For the subinterval $[0,7)$, since its operation is $(x+2) \bmod 7$, the closed-form expression can be derived as $(x+2*N) \bmod 7$.
When x is an integer, there are only 7 distinct values within the interval $[0,7)$. 
So we can also iterate from 0 to compute the value after SCC execution, finding that these 7 values have the same periodicity.
Thus we can record the execution sequence as the closed-form expression of the interval $[0,7)$.
For the subinterval $[7,10)$, its execution jumps to the $[0,7)$ interval, so the value changes before entering the periodic interval are summarized. Essentially, assuming there are $T$ values within $O$, the time complexity and space complexity of the loop summary are both $O(T)$.

When considering the outer of $O$, e.g., $(-\infty, 0)$, the value interval belongs to I-Interval of $A$. At this point, the summarization is to calculate the closed-form expression of the post-variable $x$, i.e., $x=x_0+2*N$, to make the $x$ enter $O$, i.e., $x \in [0,10)$ or to enter the T-Interval, i.e., $i<100$.

\section{Evaluation}
LoopSCC is implemented by the \texttt{python 3.12.0}, which is equipped with \texttt{z3 solver 4.13.0} for condition solving and \texttt{sympy 1.21.1} for interval computation.
We evaluate the effectiveness of LoopSCC by answering the following research questions (RQs):
\begin{itemize}
    \item RQ1: How accurate is LoopSCC in loop interpretation?
    \item RQ2: What is the effectiveness of LoopSCC in supporting practical software verification?
    \item RQ3: Can LoopSCC enhance the performance of symbolic execution?
    \item RQ4: What is the scalability of LoopSCC in real-world programs?
\end{itemize}

\subsection{RQ1: Accuracy of LoopSCC.}

\subsubsection{Benchmark}
We evaluate the summarization accuracy of LoopSCC on the \textit{C4B} \cite{carbonneaux2015compositional}, a public benchmark of loop programs that is commonly used in previous related works \cite{kincaid2017compositional} \cite{ngo2018bounded}. The \textit{C4B} has collected 36 challenging and representative loop programs from open-source software and literature. We removed loops with memory-related operations, resulting in a final set of $30$ test cases.

\subsubsection{Experimental Settings}
For each test case, we randomly generated 1k inputs and executed the original loop to obtain 1k actual outputs first. Then, we applied the comparative loop interpretation methods to compute and match the actual results, obtaining the interpretation accuracy. 
The interpretation needs to meet two restrictions to exclude the exhaustive method: \ding{182} The time for interpretation computation was limited to 5 minutes; if exceeded, it was treated as an interpretation failure. \ding{183} The generated random input is limited to a range, ensuring that the total number of iterations for a loop case is close to one million.

\subsubsection{Baselines}
Six state-of-the-art loop analysis methods have been implemented as baselines for the comparison experiments, including four abstract interpretation methods (i.e., CBMC, \\ CPAchecker, VeriAbsL and ICRA), which under-approximate or over-approximate loop behaviors, and two advanced loop summarization (concrete interpretation) methods (i.e., Proteus and WSummarizer), which preserve the original program semantics.
\begin{itemize}
    \item \textit{CBMC} \cite{ckl2004} is a classic Bounded Model Checking (BMC) \cite{biere1999symbolic} tool capable of loop analysis, which ranked first in at least one category of \textit{SV-COMP} in 2014, 2015, and 2017.
    \item \textit{CPAchecker} \cite{beyer2011cpachecker} is a well-known configurable software verification tool capable of loop analysis, which won the \textit{FalsificationOverall} category and a silver medal in the \textit{Overall} category at the \textit{SV-COMP 2024} competition.
    \item \textit{VeriAbsL} \cite{darke2023veriabsl} is a strategy prediction-based reachability verifier, which performed best in the \textit{SV-COMP 2024 ReachSafety} competition.
    \item \textit{ICRA} \cite{kincaid2017compositional} is an abstract interpretation method that combines compositional recurrence analysis with symbolic analysis for program interpretation.
    \item \textit{Proteus} \cite{xie2016proteus} is the state-of-the-art loop summarization method capable of summarizing multiple-path loops with simple cycles.
    \item \textit{WSummarizer} \cite{blicha2022summarization} is an advanced loop summarization method which is able to handle the connected cycle issue figured in the motivation examples (\S\ref{sec:motivation}).
\end{itemize}

\subsubsection{Results}

\begin{table}[ht]
  \caption{The Comparison Accuracy of Loop Summarization (\%)}
  \label{tab:interpre}
  \begin{tabular}{|c|c|c|c|c|c|c|c|} \hline
    Loop ID & CBMC & CPAchecker & ICRA & VeriAbsL & WSummarizer & Proteus & LoopSCC \\ \hline \hline
        1 & 21.7 & 10.0 & 96.0 & 9.9 & 98.8 & 100 & 100 \\ 
        2 & 21.7 & 9.5 & 43.2 & 9.5 & 0 & 0 & 100 \\ 
        3 & 24.0 & 9.4 & 9.9 & 9.4 & 0 & 0 & 100 \\ 
        4 & 24.2 & 10.8 & 39.2 & 10.7 & 82.9 & 0 & 100 \\ 
        5 & 75.3 & 52.3 & 100 & 52.3 & 81.6 & 100 & 100 \\ 
        6 & 100 & 100 & 100 & 100 & 38.2 & 100 & 100 \\ 
        7 & 55.3 & 100 & 100 & 100 & 0 & 100 & 100 \\ 
        8 & 74.2 & 100 & 100 & 100 & 0 & 100 & 100 \\ 
        9 & 72.1 & 100 & 100 & 49.1 & 0 & 100 & 100 \\ 
        10 & 69.0 & 100 & 100 & 73.6 & 80.4 & 100 & 100 \\ 
        11 & 64.5 & 78.0 & 50.3 & 78.0 & 94.7 & 100 & 100 \\ 
        12 & 68.3 & 100 & 100 & 87.4 & 0 & 100 & 100 \\ 
        13 & 77.0 & 100 & 100 & 100 & 0 & 100 & 100 \\ 
        14 & 95.3 & 100 & 100 & 100 & 100 & 100 & 100 \\ 
        15 & 67.5 & 100 & 100 & 100 & 100 & 100 & 100 \\ 
        16 & 74.5 & 100 & 100 & 100 & 100 & 100 & 100 \\ 
        17 & 46.1 & 25.4 & 100 & 25.4 & 86.7 & 100 & 100 \\ 
        18 & 58.1 & 47.2 & 100 & 100 & 0 & 100 & 100 \\ 
        19 & 82.8 & 12.1 & 100 & 100 & 100 & 100 & 100 \\ 
        20 & 75.9 & 100 & 100 & 100 & 100 & 100 & 100 \\ 
        21 & 73.1 & 100 & 100 & 77.9 & 100 & 100 & 100 \\ 
        22 & 100 & 100 & 100 & 100 & 0 & 100 & 100 \\ 
        23 & 97.9 & 95.0 & 100 & 100 & 0 & 0 & 100 \\ 
        24 & 1.9 & 2.8 & 4.4 & 0 & 0 & 0 & 100 \\ 
        25 & 83.8 & 100 & 100 & 100 & 100 & 100 & 100 \\ 
        26 & 66.2 & 51.1 & 100 & 100 & 100 & 100 & 100 \\ 
        27 & 55.5 & 51.9 & 51.9 & 51.9 & 0 & 0 & 0 \\ 
        28 & 3.1 & 90.0 & 100 & 100 & 0 & 100 & 100 \\ 
        29 & 44.9 & 27.8 & 0 & 27.8 & 0 & 0 & 0 \\ 
        30 & 96.0 & 100 & 100 & 52.1 & 52.7 & 100 & 100 \\ \hline \hline
        \#S & 30   & 30   & 29   & 29   & 16   & 23  & 28 \\
        \#H & 4    & 17   & 23   & 15   & 10   & 23  & 28 \\
        \#A & 62.3 & 69.1 & 86.0 & 72.9 & 88.5 & 100 & 100 \\ \hline
\end{tabular}
\\
\#S represents the number of test cases that can be successfully executed (accuracy is not 0\%); 
\#H  represents the number of test cases with high summarization precision (accuracy is higher than 90\%); 
\#A represents the average accuracy in successfully executed cases.
\end{table}

From Table \ref{tab:interpre}, compared to abstract interpretation tools, i.e., CBMC, CPAchecher, ICRA and VeriAbsL, LoopSCC provides a more precise interpretation of loops, achieving a  100\% average accuracy. Compared to advanced loop summarization tools, i.e., WSummarizer and Proteus, LoopSCC can handle a wider range of loop types. Out of 30 test cases, LoopSCC successfully handled 28, which is more than Proteus and WSummarizer by 5 and 12, respectively.
Especially, LoopSCC can correctly summarize multi-path loops with \textit{inductiveness traps} and \textit{connected cycles}, as demonstrated by the case \texttt{custom\_4} (see Fig. \ref{custom4}), where Proteus and WSummarizer fail to handle them. 
The interpretation tools handle a wider variety of loop types, but lose partial program semantics. For instance, from the success number, the worst-performing interpretation tools (ICRA and VeriAbsL) are unable to handle only one test case, which is higher than the optimal loop summarization method. The execution success number reflects application generalization, indicating that interpretation tools generally cover a broader range of case types compared to summarization methods.
However, from the test accuracy, the highest average accuracy achieved by the interpretation tools is 86.0\% which is even lower than the worst summarization method (WSummarizer). By further analysis, we found that running WSummarizer with the same input twice is possible to result in different output values. This indicates that WSummarizer lacks appropriate constraint expressions in its implementation, which leads to inadequate restrictions on variable values.

\begin{wrapfigure}{r}{0.35\textwidth}
  \footnotesize
  \fbox{\begin{minipage}{0.3\textwidth}
   \texttt{if y >= 100:}\\
   \hspace*{0.5cm}\texttt{n = floor((y - 100) / 100)}\\
   \hspace*{0.5cm}\texttt{y = y - 100 * n}\\
  \end{minipage}}
\caption{Linear Program of Inner Loop in \texttt{t27}.}\label{fig:linear-program}
\end{wrapfigure}

LoopSCC only failed to execute in 2 cases, which also caused failures in Proteus and WSummarizer.
To deep analyze the 2 cases, one is a nested loop \texttt{t27} (see Fig. \ref{t27}) with $y=y-100$.
After eliminating inner loop structures, LoopSCC produces a linear program (see Fig. \ref{fig:linear-program}) that includes the \texttt{floor} computation, making it unable to be summarized. The other case contains a \textit{multivariate recursion} operation in the loop, i.e., $x_n=y_{n-1}$, $y_n=x_{n-1}-1$, which has not been realized in the LoopSCC.

\begin{figure}[htbp]
\begin{minipage}[b]{.3\textwidth}
\begin{lstlisting}
while i < 100:
    if flag:
        if x > 5:
            x -= 5
            i += 3
        else:
            x += 2
            i += 7
    else:
        x -= 7
        flag = 1
\end{lstlisting}
\subcaption{custom\_4 (ID: 4)}\label{custom4}
\end{minipage}
\hfill
\begin{minipage}[b]{.3\textwidth}
\begin{lstlisting}
while n < 0:
    n = n + 1
    y = y + 1000
    while y >= 100:
        y = y - 100
\end{lstlisting}
\subcaption{t27 (ID: 27)}\label{t27}
\end{minipage}
\hfill
\begin{minipage}[b]{.3\textwidth}
\begin{lstlisting}
while x > 0:
    x = x - 1
    t = x
    x = y
    y = t
\end{lstlisting}
\subcaption{t30 (ID: 29)}\label{t30}
\end{minipage}
\caption{Representive Test Cases in RQ1.}
\label{experiment1}
\end{figure}

\subsection{RQ2: Support of Software Verification}

\subsubsection{Benchmark}
The benchmark is collected from \textit{SV-COMP 2024} \footnote{SV-COMP 2024: \url{https://gitlab.com/sosy-lab/benchmarking/sv-benchmarks/-/blob/svcomp24-final}}, one of the most well-known competitions in the field of software verification. Aiming at loop summarization, we extracted 6  programs branches related to loops (i.e., \texttt{loop-acceleration}, \texttt{loop-crafted}, \texttt{loop-new},\texttt{loops}, \texttt{loop-zilu} and \texttt{loop-simple}) from the original benchmark. 
We removed loops with memory-related operations, resulting in a final total of 107 test cases for evaluation.

\subsubsection{Experimental Settings}
Software verification aims to ensure that software systems function as intended and meet specified requirements. In practice, the software verification employs specific \texttt{assert} statements to determine whether the target property within certain program path holds.
To evaluate the effectiveness of LoopSCC in supporting the software verification, we have added new functions to LoopSCC, enabling it to check both assertions within the loops and those outside the loops:
\ding{182} For assertions within the loops, LoopSCC identifies all paths containing properties to be verified, and converts the assert conditions' negation into expressions of the $sp.Pre$ for these paths through variable substitution. 
After that, LoopSCC computes the loop summarization.
For any assertion, if there exist values that satisfy both the loop summarization and the negated condition of the assertion, then this assertion is incorrect.
\ding{183} For assertions outside the loops, we directly append the negation of the assert condition to the end of the loop summarization and check whether the new summarization can be satisfied to determine if the property holds. 
Specifically, for cases where the loop condition is uncertain (i.e., the condition value is set as \texttt{nondet}), LoopSCC introduces two variables $x$ and $y$ to represent arbitrary iteration number and sets an explicit loop condition $x<y$ to replace this uncertainty.
This behavior does not change the semantics.

\subsubsection{Baselines}
In the 6 baselines figured in RQ1, only CBMC, CPAchecker, VeriAbsL and Proteus can be successfully applied in the software verification task and achieve promising results. Especially, VeriAbsL won the Top 1 in the competition of \textit{SV-COMP 2024 ReachSafety}. 
Therefore, the 4 methods, i.e., CBMC, CPAchecker, VeriAbsL and Proteus, are implemented as comparative baselines in this experiment.

\subsubsection{Results}

\begin{figure}[ht]
\centering
\begin{tikzpicture}[scale=0.8]
\begin{axis}[
bar width=15pt,
    nodes near coords,
    symbolic x coords={CBMC, CPAchecker, VeriAbsL, Proteus, LoopSCC},
    xtick=data,
    ymin=0, ymax=100,
    ylabel={Accuracy (\%)},
    xlabel={Comparative Tools},
    x tick label style={font=\footnotesize},
]
\addplot [ybar, fill=blue!60, postaction={pattern=north east lines}]
	coordinates {(CBMC,19.6) (CPAchecker,60.8)
		 (VeriAbsL,75.7) (Proteus,72.0) (LoopSCC,86.0)};
\end{axis}
\end{tikzpicture}
\caption{Accuracy Comparison in Software Verification.}
\label{fig:experiment2}
\end{figure}
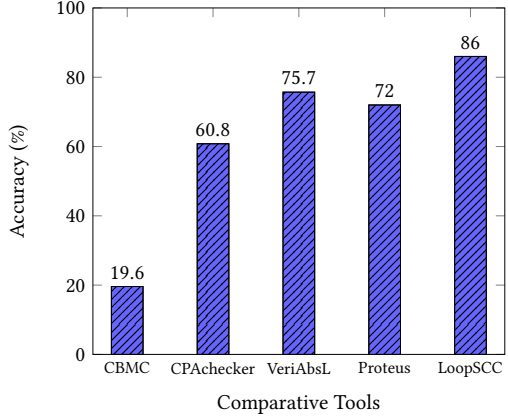

From Fig. \ref{fig:experiment2}, LoopSCC successfully verified 86\% of the test cases, outperforming the best-performing baseline, i.e., VeriAbsL, by 10.3\%.
From further analysis, we found that the verification errors in LoopSCC are primarily attributed to the incomplete implementation of complex operations, resulting in the closed-form expressions being hard to represent. For instance, the test cases \texttt{phases\_2-1}, \texttt{phases\_2-2} and \texttt{underapprox} involve in operations of \textit{square} or \textit{division}. Notably, Proteus, serving as a state-of-the-art summarization method, performs slightly worse than VeriAbsL. This is because VeriAbsL is a hybrid approach that integrates multiple advanced summarization methods and utilizes machine learning to select the most suitable one.

\subsection{RQ3: Enhancement of Symbolic Execution}

\subsubsection{Experimental Settings}
LoopSCC uses condition expressions derived from flow-based SPath analysis to compute loop summarization, which can be integrated with existing program analysis techniques to enhance its analysis performance. 
To evaluate the effectiveness of LoopSCC in such application scenarios, we take \textit{symbolic execution}, a well-known software analysis technique as a typical target to enhance. 
To this end, we take five well-known algorithm programs as the test cases, i.e., \textit{integer division}, \textit{factorization}, \textit{GCD}, \textit{LCM} and \textit{square root} \cite{de2016polynomial}. 
A common scenario for symbolic execution is verifying feasible inputs that satisfy data integrity through reverse condition-based reasoning.
In this experiment, we will simulate this scenario.

We then use \textit{KLEE} \cite{cadar2008klee}, an advanced symbolic execution tool, to run the test cases and get the execution results and time costs. To ensure that \textit{KLEE} executes without trapping in an infinite loop, we use the \verb|klee_assume| command to limit the range of values for the input variables of each test case. In detail, we restrict the range of values to $[0, 10k]$ for single variables and to $[0, 100]$ for dual variables.

\subsubsection{Results}

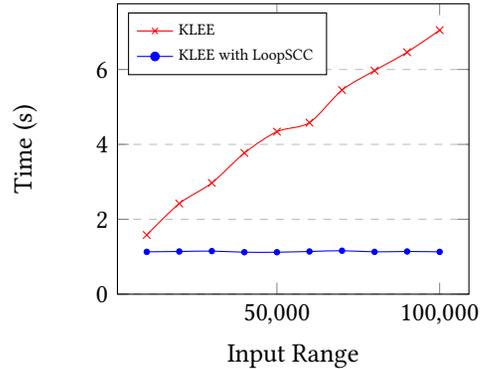
\begin{figure}[ht]
\begin{minipage}[c]{.5\textwidth}
\centering
\begin{tabular}{|c|c|c|}
\hline
Tools    & KLEE    & KLEE+LoopSCC \\
\hline \hline
cohendiv & 14 m 56 s  & 1.15 s         \\
fermat   & 04 h 49 m  & 1.32 s         \\
GCD      & 19 m 56 s  & 1.24 s         \\
LCM      & 26 m 16 s  & 1.31 s         \\
sqrt     & 1.79 s     & 1.13 s         \\
\hline
\end{tabular}
\subcaption{Time Costs with/without LoopSCC.}
\end{minipage}
\hfill
\begin{minipage}[c]{.45\textwidth}
\centering
\begin{tikzpicture}
\begin{axis}[
    width=\textwidth,
    scaled ticks=false,
    tick label style={/pgf/number format/fixed},
    xtick={0,50000,100000},
    ymin=0,
    legend style={
        font=\tiny
    },
    legend pos=north west,
    legend image post style={mark size=2pt, scale=0.8},
    legend cell align={left},
    ymajorgrids=true,
    grid style=dashed,
    ylabel={Time (s)},
    xlabel={Input Range},
]
\addplot [smooth,color=red,mark=x,]
	coordinates {(10000,1.58) (20000,2.42)
		 (30000,2.97) (40000,3.77) (50000,4.34)
      (60000,4.58) (70000,5.45) (80000,5.97)
      (90000,6.46) (100000,7.05)};
\addplot [smooth,color=blue,mark=*,mark size=1pt]
	coordinates {(10000,1.13) (20000,1.14)
		 (30000,1.15) (40000,1.12) (50000,1.12)
      (60000,1.14) (70000,1.16) (80000,1.13)
      (90000,1.14) (100000,1.13)};
\legend{KLEE,KLEE with LoopSCC}
\end{axis}
\end{tikzpicture}
\subcaption{Time Costs and Input Ranges in \texttt{sqrt}.}
\end{minipage}
\caption{Symbolic Execution Supported by LoopSCC.}
\label{fig:experiment3}
\end{figure}

From Fig. \ref{fig:experiment3}(a), LoopSCC presents significant loop acceleration capabilities, reducing symbolic execution time from initially several minutes or even hours to around 1 second. Additionally, it can be observed that regardless of the original execution time, the analysis time for symbolic execution with LoopSCC is roughly consistent. This is because loop summarization addresses the most time-consuming part of the program execution, i.e., loop structure, making the execution time of the remaining code negligible.
Furthermore, we found that for the original symbolic execution, as the range of input variable values increases, the execution time grows significantly. 
Fig. \ref{fig:experiment3}(b) demonstrates the relationship between time cost and input range in \texttt{sqrt} algorithm.  It can be seen that the loop summary computed by LoopSCC are unaffected by the number of iterations caused by the input range, maintaining relatively stable execution times.

\subsection{RQ4: Scalability on Real-World Loops}
\subsubsection{Experimental Settings}
To explore the scalability of LoopSCC in handling real-world loops, we conduct a systematic investigation on three open-source utility programs with large-scale code, i.e.,  \texttt{Bitcoin}, \texttt{musl} and \texttt{Z3}. 
We first compile the target program into an \texttt{LLVM IR} file and use the \texttt{LoopAnalysisManager} analyzer to retrieve all loops within the programs.
Then all the code operations are identified to remove the loops with memory-related operations which is out of our scope. At last, a total of 7,406 effective loops are collected.
Afterwards, we use the \verb|LoopSimplify| \texttt{pass} to simplify these loops to obtain all paths within the loops. This allows us to check the jumps between paths to determine the order of SCCs contained in each loop.
For the loops with high-order SCCs, we continue to explore the existence of oscillatory intervals within these SCCs. In this process, the \verb|reduce_inequalities| function of \textit{sympy} library is exploited to compute the oscillatory intervals following the method described in \S\ref{subsec:high-order}.
In handling the oscillatory intervals with finite values, the average time cost of LoopSCC achieves a low overhead, where the average time cost is less than 1 second and the memory usage is less than 300 KB.

\begin{table}[ht]
  \caption{Loops in Real World.}
  \label{tab:investigate}
  \begin{tabular}{|c|c|c|c|c|} \hline
  Program & Total & Without High-order SCC & With High-order SCCs & \#FOIs \\ \hline \hline
  Bitcoin & 1437 & 816(56.8\%) & 621(43.2\%) & 548(88.2\%) \\
  musl & 241 & 165(68.5\%) & 76(31.5\%) & 69(90.8\%) \\
  Z3 & 5728 & 3706(64.7\%) & 2022(35.3\%) & 1903(94.1\%) \\ \hline \hline
  Total & 7406 & 4687(63.3\%) & 2719(36.7\%) & 2520(92.7\%) \\ \hline
  \end{tabular}
  \\ \#FOI: finite oscillatory interval where the number of values are less than 1 million.
\end{table}

\subsubsection{Results}
From Table \ref{tab:investigate}, most loops (63.3\%) execute without any high-order SCC, as many of the loops are simply single-branch for loops. Among the 36.7\% of loops that contain high-order SCCs, 92.7\% have a finite oscillatory interval and can be summarized by LoopSCC.
However, 1,169 of these programs contain complex nested loops, making it challenging to directly derive closed-form expressions for the recurrence operations. Ultimately, LoopSCC can successfully summarize 6,038 (81.5\% ) of the collected loops, demonstrating excellent scalability in real-world programs.

\section{Limitations}\label{sec:threats}

\subsection{Execution without Periodicity}
Our method requires execution periodicity in the summarization of high-order SCCs.
The periodic execution exists only when the number of values within oscillatory interval are finite. 
However, when the oscillatory interval is infinite, the loop execution may not present excepted periodicity, 
making LoopSCC invalid.
For example, the loop shown in Fig. \ref{2orderscc}(b) contains a 2-order SCC.
If the type of variable $x$ is real number, the oscillatory interval becomes infinite, which prevents LoopSCC from summarizing the loops with high-order SCCs.

\subsection{Inductiveness Trap in Nesting Eliminating}
LoopSCC utilizes an inside-out transformation to convert the nested loops into non-nested loops.
Note that the loop summary typically contains iteration variables, which is possible to complicated the operation expressions of summarization even if the loop operations are quite simple. 
Then LoopSCC will suffer from \textit{inductiveness trap} problem when transforming nested loops as traditional methods.
As such, after eliminating the inner loop, the operations of outer loop are hard to produce closed-form expressions, leading to the failure of summarization. 
As a result, LoopSCC can only summarize a subset of nested loops.
A possible solution is to use program synthesis methods to synthesize the summaries of the inner loops with the outer loop program, which is our focus in the future work. 

\section{Related Work}\label{sec:rw}

\subsection{Loop Summarization}
Loop summarization aims to generate a set of constraint expressions to represent the mapping between loop inputs and outputs, which can be used to directly replace loops in program analysis. 
Techniques for loop summarization can broadly be categorized into two types: 
those based on symbolic execution \cite{saxena2009loop} \cite{godefroid2011automatic} \cite{strejvcek2012abstracting} \cite{kapus2019computing} 
and those based on data-flow analysis \cite{xie2016proteus} \cite{blicha2022summarization} \cite{xie2015slooper}.
Currently, most of the loop summarization efforts target at single-branch loops or multi-branch loops without path jump, since complex structures with interconnected branches are hard to analyze. For instance, LESE\cite{saxena2009loop}, APLS\cite{godefroid2011automatic} and APC\cite{strejvcek2012abstracting} are such summarization efforts which derive loop summaries by detecting linear relations among induction variables during symbolic execution.
Some works attempt to address the multi-branch challenge by analyzing the execution patterns of loop paths.
Proteus \cite{xie2016proteus} analyzes the jump relations between paths and introduces a PDA automaton model. When the jumps between multiple paths produce a simple cycle, Proteus propagates variable expressions along the simple cycle to derive the symbolic expression for the entire cycle. Such design enable Proteus to summarize multi-branch loops that contain simple cycles.
WSummerizar\cite{blicha2022summarization} combines \textit{loop unwinding} with the Proteus method, enabling it to summarize specific multi-branch loops that contain \textit{connected cycles}.
Apart from that, other efforts aim to handle specific operations within the loop summarization.
Xie et al. \cite{xie2015slooper} adapted the Proteus method for string operations and proposed the S-Looper approach. Similarly, Kapus et al. \cite{kapus2019computing} developed a loop summarization method focused on string operations, using CounterExample-Guided Inductive Synthesis (CEGIS) \cite{solar2006combinatorial}. 
These methods normally translate a loop program into a simpler loop consisting solely of primitive operations, e.g., pointer increment, and standard string operations,  e.g., string copy. 
\textit{Compared to traditional works, LoopSCC employs improved SCC-based flow analysis, to summarize diverse types of multi-branch loops, alleviating the challenge of connected cycles and inductiveness traps.}

\subsection{Program Abstraction}
Program abstraction aims to describe a target structure within the program through specific features or formulas. 
For instance, \textit{invariant generation}, serving as a loop analysis approach, is commonly employed to find correct constraint relation of loop, which can be categorized into two types, i.e., white-box and black-box. 
For white-box based invariant generation, representative works include counter-example guided abstraction refinement (GEGAR \cite{clarke2000counterexample}) and predicate abstraction \cite{ball2001automatic} \cite{flanagan2002predicate} \cite{ball2002slam} \cite{henzinger2002lazy} , 
Craig interpolation \cite{mcmillan2006lazy} \cite{mcmillan2003interpolation}, 
and IC3/PDR based abstraction \cite{bradley2011sat}.
Black-box based techniques include advanced machine learning based feature extraction for target loop structures \cite{garg2014ice} \cite{si2020code2inv} \cite{pmlr-v202-pei23a}.
Compositional analysis (CRA) \cite{Farzan_Kincaid_2015} serving as another program abstraction method, has gained increasing focuses recently, which adopts a \textit{divide-and-conquer} strategy, to synthesize the abstractions from each sub-programs.
There are different variations of CRA methods, such as  predicate abstraction based 
program composition \cite{kroening2008loop} and  numerical abstraction based approximation for program behavior \cite{muller2004precise}\cite{monniaux2009automatic}.  
Kincaid et al. \cite{kincaid2017compositional} prososed ICRA  method, an extension of CRA, which adapts context-sensitive inter-procedural programs. Specifically,  ICRA integrates the Newton iteration method \cite{esparza2010newtonian} to handle non-linear variable operations. \textit{Different from existing program abstraction efforts, LoopSCC aims to provide expression-based summaries that precisely interpret loops, which supporting further program analysis.}

\section{Conclusion}
This paper proposes an SCC-based flow analysis method for complex multi-branch loop structures summarization. In this process, specific branch-based graph is extracted to facilitate recurrence analysis for condition relationships of loop iterations. Especially, for high-order SCCs within the target loops,  oscillatory interval is proposed and quantitatively analyzed to transform indeterminate iteration period into a determinate one. Thus, complex loop structures with connected cycle and inductiveness trap can be alleviated.
Compared to state-of-the-art methods, the proposed approach covers a wider range of loop types and achieves higher summarization accuracy. Additionally, it presents outstanding scalability to real-world loop programs.


\bibliographystyle{ACM-Reference-Format}
\bibliography{main}


\end{document}